\documentclass[journal,draftcls, onecolumn]{IEEEtran}
\usepackage[dvipdfmx]{graphicx}
\usepackage{amsmath}
\usepackage{amssymb}
\usepackage{amsfonts}
\usepackage{mathrsfs}
\usepackage{theorem}
\usepackage{color}
\DeclareMathAlphabet{\bm}{OML}{cmm}{b}{it}
\theorembodyfont{\rmfamily}
\newtheorem{theorem}{Theorem}
\newtheorem{lemma}{Lemma}
\newtheorem{definition}{Definition}
\newtheorem{corollary}{Corollary}

\newtheorem{proposition}{Proposition}

\newcommand{\qed}{\hfill \IEEEQED}

\newcommand{\markov}{ - \!\!\circ\!\! - }

\newcommand{\bol}[1]{\mathbf{#1}}
\newcommand{\rom}[1]{\mathrm{#1}}
\newcommand{\san}[1]{\mathsf{#1}}

\newcommand{\Pe}{\rom{P}_{\rom{e}}}
\newcommand{\Pc}{\rom{P}_{\rom{c}}}

\begin{document}

\title{Second-Order Region for Gray-Wyner Network}

\author{Shun Watanabe
\thanks{S.~Watanabe is with the Department of Computer and Information Sciences, Tokyo University of Agriculture and Technology, Japan, E-mail:shunwata@cc.tuat.ac.jp.}
}


\maketitle
\begin{abstract}
The coding problem over the Gray-Wyner network is studied from the second-order coding rates perspective. 
A tilted information density for this network is introduced in the spirit of Kostina-Verd\'u, and, under a certain regularity condition, 
the second-order region is characterized in terms of the variance of this tilted information density and the tangent vector of the first-order region. 
The second-order region is proved by the type method: the achievability part is proved by the type-covering argument,
and the converse part is proved by a refinement of the perturbation approach that was used by Gu-Effros to show the strong converse of the
Gray-Wyner network. This is the first instance that the second-order region is characterized for a multi-terminal problem where
the characterization of the first-order region involves an auxiliary random variable.  
\end{abstract}

\section{Introduction}

We study the coding problem over the Gray-Wyner network \cite{GraWyn:74} from the second-order
coding rates perspective. The study of the second-order coding rates has attracted significant interest 
in recent years since it gives a good approximation for the finite blocklength performance of certain coding systems \cite{hayashi:09, polyanskiy:10}. 
The second-order coding rates for point-to-point systems are quite well-understood 
\cite{Strassen:62, hayashi:09, polyanskiy:10, hayashi:08, KosVer:15, nomura:13, kostina:12, ingber:11,KosVer:13, WanIngKoc:11, KosVer:14, WatHay14}.
On the other hand, the extension of the second-order analysis to multi-terminal problems is rather immature; some problems 
are solved completely \cite{tan:12, NomHan:14, tomamichel:13, LeTanMot:15, ScaTan:13, HayTyaWat14, TyaVisWat:15, TyaShaVisWat:15}, 
but only achievability bounds are known for other problems \cite{tan:12, MolLan:12, HuaMou:12, ScaMarFab:15, watanabe:13e, YasAreGoh:13, Sca:13}.
See \cite{Tan:book} for further review of existing results on the second-order analysis. 


The Gray-Wyner network is described in Fig.~\ref{Fig:GW-network}.
The network consists of one encoder and two decoders. The encoder and both the decoders 
are connected by the common channel, and each decoder is also connected to the encoder by
its own private channel. Then, the goal for each decoder is to almost losslessly reproduce one part
of correlated sources, and we are interested in the optimal trade-off among the rates of the three channels. 
The information theoretic chanracterization of achievable rate triplets was derived in \cite{GraWyn:74}, and, as is typical for multi-terminal problems (cf.~\cite{elgamal-kim-book}),
it involves an auxiliary random variable, which makes the second-order analysis of this problem non-trivial. 


\subsection{Contributions}

We characterize the second-order region of the Gray-Wyner network under a certain regularity condition.\footnote{Because of the 
regularity condition, our result cannot be applied to singular points on the boundary of the first-order region, i.e., the 
boundary points where the first-order region cannot be differentiated.} 
For that purpose, we introduce a tilted information density for this network in the spirit of Kostina-Verd\'u \cite{kostina:12}.
Then, the second-order region is characterized in terms of the variance of this tilted information density and the tangent vector of the first-order region. 
Since the first-order region of the Gray-Wyner  network is characterized by an auxiliary random variable, 
the tilted information density is defined by using that auxiliary random variable.
In general, there is no guarantee that an optimal test channel is unique, and 
more than one optimal test channel may exist.
However, we show that the tilted information density is uniquely defined
irrespective of the choice of optimal test channels. 
Also, we show some other properties of the tilted information density. 

In \cite{GraWyn:74}, the plane where the sum of the three rates coincide with the joint entropy of the correlated sources was called the
Pangloss plane, and it gained a special attention since there is no sum-rate loss compared to cooperative decoding schemes on this plane. 
When the first-order rates are on the Pangloss plane, as an illustration of our main result, we show a simple expression of the second-order region. 
Interestingly, the sum constraint of the second-order rates coincide with that can be achieved by cooperative decoding schemes; 
this means that there is no sum-rate loss compared to cooperative decoding schemes even up to the second-order.  

In the proof of the second-order region, we use the type method.
The achievability part is proved by an application of the type covering argument (cf.~\cite{YuSpe:93} and \cite[Chapter 9]{csiszar-korner:11}).
For the converse part, we refine the perturbation approach that was used by Gu-Effros \cite{GuEff:09, GuEff:11} to show the strong converse of the
Gray-Wyner network. By these argument, we first derive an upper bound and a lower bound on the error probability in terms of 
a probability of a certain function of the joint type. Then, we approximate that probability by using the central limit theorem. 

When we use the type method for the second-order analysis, say the rate-distortion problem, we take a derivative of the rate-distortion function
with respect to the source distribution, and the second-order rate is characterized in terms of the variance of that derivative (cf.~\cite{ingber:11}). 
Then, we can show that that characterization coincides with the variance of the $d$-tilted information introduced in \cite{kostina:12}.
In this paper, we consider a slightly different argument. When we take a derivative of a certain function of a distribution, we have to
extend the domain of that function to the outside of the probability simplex (cf.~\cite[Appendix A]{NoIngWei:15}). In order to circumvent 
such an extension, we consider a different parameterization of the probability simplex, which is often used in information geometry \cite{amari-nagaoka}.
Then, we take a derivative of the function with respect to that parameter. Also, instead of introducing the variance of the derivative, we directly 
characterize the second-order region in terms of the variance of the tilted information density.

\subsection{Paper Organization}

The rest of the paper is organised as follows:
In Section \ref{section:problem}, we introduce our notation, and recall the problem formulation of
the Gray-Wyner network. In Section \ref{section:tilted}, we introduce the tilted information density 
for the Gray-Wyner network, and investigate its properties. 
Then, in Section \ref{section:coding-theorem}, we show our second-order coding theorem and its proof.
In Section \ref{section:pangloss}, we further investigate the Pangloss plane. 
We conclude the paper with some discussions in Section \ref{section:discussion}.


\begin{figure}[tb]
\centering{
\begin{minipage}{.5\textwidth}
\includegraphics[width=\textwidth]{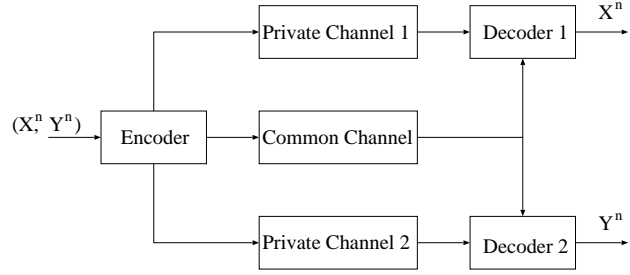}
\caption{A description of the Gray-Wyner network.}
\label{Fig:GW-network}
\end{minipage}
}
\end{figure}

\section{Problem Formulation} \label{section:problem}

In this section, we introduce our notations and recall the Gray-Wyner network \cite{GraWyn:74}.

\subsection{Notations} \label{subsection:notation}

Random variables (e.g.~$X$) and their realizations (e.g.~$x$) are in capital and lower 
case, respectively. All random variables take values in some finite alphabets which are
denoted in calligraphic font (e.g.~${\cal X}$). The cardinality of ${\cal X}$ is denoted as $|{\cal X}|$.
Let the random vector $X^n = (X_1,\ldots,X_n)$ and similarly for a realization 
$\bm{x} = (x_1,\ldots,x_n)$. For information theoretic quantities, 
we follows the same notations as \cite{csiszar-korner:11}; e.g.~the 
entropy and the mutual information are denoted by $H(X)$ and $I(X \wedge Y)$, respectively.
Also, the expectation and the variance are denoted by $\san{E}[\cdot]$ and $\san{V}[\cdot]$ respectively.
$\san{Q}(t) = \int_t^\infty \frac{1}{\sqrt{2\pi}} e^{- \frac{u^2}{2}} du$ is the upper tail probability of the standard normal
distribution; its inverse is denoted by $\san{Q}^{-1}(\varepsilon)$ for $0 < \varepsilon < 1$. 

The set of all distribution on ${\cal X}$ is denoted by ${\cal P}({\cal X})$. 
The set of all channels from
${\cal X}$ to ${\cal Y}$ is denoted by ${\cal P}({\cal Y}|{\cal X})$.
We will also use the method of types \cite{csiszar-korner:11}. 
For a given sequence $\bm{x}$, its type is denoted by $\san{t}_{\bm{x}}$.
The set of all types on ${\cal X}$ is denoted by ${\cal P}_n({\cal X})$, and the 
set of all conditional types is denoted by ${\cal P}_n({\cal Y}|{\cal X})$. For a given type $P_{\bar{X}} \in {\cal P}_n({\cal X})$,
the set of all sequences with type $P_{\bar{X}}$ is denoted by ${\cal T}_{\bar{X}}^n$.
For a given joint type $P_{\bar{X}\bar{Y}}$ and a sequence $\bm{x} \in {\cal T}_{\bar{X}}^n$, the 
set of all sequences whose joint type with $\bm{x}$ is $P_{\bar{X}\bar{Y}}$ is denoted by ${\cal T}_{\bar{Y}|\bar{X}}^n(\bm{x})$.
For type $P_{\bar{X}}$ and joint type $P_{\bar{X}\bar{Y}}$, we use notations $H(\bar{X})$ and $I(\bar{X} \wedge \bar{Y})$,
where the random variables are distributed according to those type and joint type.

For a given distribution $P_X$, its support is denoted by $\mathtt{supp}(P_X)$. In latter sections, we will
differentiate a certain function of distributions around a given joint distribution $P_{XY}$, which may not have full support. 
For that purpose, it is convenient to introduce a
parametrization for distribution $P$ that has the same support as $P_{XY}$.\footnote{In the literature \cite{ingber:11}
(see also \cite{NoIngWei:15}),
the probability simplex is embedded into the Euclidian space, and the parameterization
on that Euclidian space is used. However in this paper, we regard the probability simplex as 
a manifold (cf.~\cite{amari-nagaoka}), and we consider a parameterization
that is different from the literature so that we do not have to extend the domain of a certain function to outside the probability simplex.} 
Let $m = \mathtt{supp}(P_{XY})$; without loss of generality, we assign 
$1$ through $m$ to elements in $\mathtt{supp}(P_{XY})$. Then, parameter $\theta( P) \in \mathbb{R}^{m-1}$ is
defined as $\theta_i = P(i)$ for $i = 1,\ldots,m-1$; apparently it holds $P(m) = 1- \sum_{i=1}^{m-1} \theta_i$. 
The distribution corresponding to parameter $\theta$ 
is denoted by $P_\theta$. 

\subsection{Gray-Wyner Network}

In this section, we recall the lossless source coding problem over the Gray-Wyner network (see Fig.~\ref{Fig:GW-network}).
Let us consider a correlated source $(X,Y)$ taking values in ${\cal X}\times {\cal Y}$ and 
having joint distribution $P_{XY}$. We consider a block coding of length $n$.
A coding system consists of three encoders
\begin{align}
& \varphi^{(n)}_{0}: {\cal X}^n \times {\cal Y}^n \to {\cal M}_0^{(n)}, \\
& \varphi^{(n)}_1: {\cal X}^n \times {\cal Y}^n \to {\cal M}^{(n)}_1, \\
& \varphi^{(n)}_2: {\cal X}^n \times {\cal Y}^n \to {\cal M}^{(n)}_2, 
\end{align}
and two decoders
\begin{align}
& \psi^{(n)}_1: {\cal M}^{(n)}_0 \times {\cal M}^{(n)}_1 \to {\cal X}^n, \\
& \psi^{(n)}_2: {\cal M}^{(n)}_0 \times {\cal M}^{(n)}_2 \to {\cal Y}^n.
\end{align}
The message encoded by $\varphi^{(n)}_{0}$ is sent over the common channel, and received by
both the decoders; the message encoded by $\varphi^{(n)}_i$ is sent over the private channel
to $i$th decoder, where $i=1,2$. The first decoder is required to reproduce $X^n$ almost losslessly,
while the second decoder is required to reproduce $Y^n$ almost losslessly. 
In the following, we omit the blocklength $n$ when it is obvious from the context.
For $(X^n, Y^n) \sim P$, the error probability of code $\Phi_n = (\varphi_{0}, \varphi_{1}, \varphi_{2}, \psi_1, \psi_2)$ is defined as 
\begin{align}
\Pe(\Phi_n|P) := \Pr\bigg( (\psi_1(\varphi_{0}(X^n,Y^n), \varphi_{1}(X^n,Y^n)), 
 \psi_2(\varphi_{0}(X^n,Y^n), \varphi_{2}(X^n,Y^n))) \neq (X^n, Y^n) \bigg).
\end{align}
Then, the correct probability of the code is defined as 
\begin{align}
\Pc(\Phi_n|P) := 1  - \Pe(\Phi_n|P).
\end{align}
In the following, we are particularly interested in the case where $P$ is a product distribution $P_{XY}^n$,
i.e., $(X^n,Y^n)$ is an i.i.d. sequence.

\begin{definition}[First-Order Region]
The rate triplet $(r_0,r_1,r_2) \in \mathbb{R}_+^3$ is defined to be achievable if there exists a 
sequence of code $\{ \Phi_n \}_{n=1}^\infty$ such that 
\begin{align}
\limsup_{n\to \infty} \frac{1}{n} \log |{\cal M}_0^{(n)}| &\le r_0, \\
\limsup_{n\to \infty} \frac{1}{n} \log |{\cal M}_1^{(n)}| &\le r_1, \\
\limsup_{n\to \infty} \frac{1}{n} \log |{\cal M}_2^{(n)}| &\le r_2,
\end{align}
and
\begin{align}
\lim_{n \to \infty} \Pe(\Phi_n|P_{XY}^n) = 0.
\end{align}
Then, the achievable region ${\cal R}_{\mathtt{GW}}(P_{XY})$ is defined as the 
set of all achievable rate triplets. 
\end{definition}

The first-order region ${\cal R}_{\mathtt{GW}}(P_{XY})$ is characterized in \cite{GraWyn:74}.
Let ${\cal R}_{\mathtt{GW}}^*(P_{XY})$ be the set of all rate triplets $(r_0,r_1,r_2)$ such that there exists 
a test channel $P_{W|XY}$ with $|{\cal W}| \le |{\cal X}||{\cal Y}| + 2$ such that
\begin{align}
r_0 &\ge I(W \wedge X,Y), \\
r_1 &\ge H(X|W), \\
r_2 &\ge H(Y|W).
\end{align}

\begin{proposition}[\cite{GraWyn:74}]
It holds that\footnote{In fact, the cardinality bound was not shown in \cite{GraWyn:74}, but it can be
proved by the support lemma \cite{csiszar-korner:11} (see also \cite{KamAna:10}).} 
\begin{align}
{\cal R}_{\mathtt{GW}}(P_{XY}) = {\cal R}^*_{\mathtt{GW}}(P_{XY}).
\end{align}
\end{proposition}

In this paper, we are interested in the second-order region. We follow the second-order formulation in \cite{NomHan:14}.
\begin{definition}[Second-Order Region]
For a boundary point $(r_0^*, r_1^*, r_2^*)$ of ${\cal R}_{\mathtt{GW}}(P_{XY})$ and $0 < \varepsilon < 1$, 
the rate triplet $(L_0,L_1,L_2) \in \mathbb{R}^3$ is defined to be $(\varepsilon,r_0^*, r_1^*, r_2^*)$-achievable if
there exists a sequence of code $\{\Phi_n \}_{n=1}^{\infty}$ such that
\begin{align}
\limsup_{n\to\infty} \frac{\log |{\cal M}_0^{(n)}| - n r_0^*}{\sqrt{n}} &\le L_0, \label{eq:definition-second-achievability-0} \\
\limsup_{n\to\infty} \frac{\log |{\cal M}_1^{(n)}| - n r_1^*}{\sqrt{n}} &\le L_1, \label{eq:definition-second-achievability-1} \\
\limsup_{n\to\infty} \frac{\log |{\cal M}_2^{(n)}| - n r_2^*}{\sqrt{n}} &\le L_2, \label{eq:definition-second-achievability-2}
\end{align}
and 
\begin{align}
\limsup_{n\to\infty} \Pe(\Phi_n|P_{XY}^n) \le \varepsilon.
\end{align}
Then, the $(\varepsilon, r_0^*, r_1^*, r_2^*)$-achievable region ${\cal L}_{\mathtt{GW}}(\varepsilon; r_0^*,r_1^*,r_2^*)$ is defined
as the set of all $(\varepsilon, r_0^*, r_1^*, r_2^*)$-achievable rate triplets.
\end{definition}

In contrast to first-order rates, second-order rates may be negative even though they are conventionally called ``rates".

\section{Tilted Information Density} \label{section:tilted}

In this section, we introduce the tilted information density for the Gray-Wyner network
in the spirit of \cite{kostina:12}. The tilted information density plays an important role to
characterize the second-order region ${\cal L}_{\mathtt{GW}}(\varepsilon; r_0^*,r_1^*,r_2^*)$ in the next section.

Given $r_1,r_2 > 0$, let 
\begin{align} \label{eq:optimal-r0-function}
R(r_1,r_2|P_{XY}) &:= \min\big\{ r_0 : (r_0,r_1,r_2) \in {\cal R}_{\mathtt{GW}}^*(P_{XY}) \big\} \\
&= \min\big\{ I(W \wedge X,Y): |{\cal W}| \le |{\cal X}||{\cal Y}| +2, r_1 \ge H(X|W), r_2 \ge H(Y|W) \big\}.
\end{align}
Since ${\cal R}_{\mathtt{GW}}^*(P_{XY})$ is a convex region, an optimal test channel satisfies 
the conditions $ r_1 \ge H(X|W)$ and $r_2 \ge H(Y|W)$ with equality unless $R(r_1,r_2|P_{XY})=0$.

Throughout the paper, we assume that ${\cal R}_{\mathtt{GW}}^*(P_{XY})$ is smooth at 
a boundary point $(r_0^*,r_1^*,r_2^*)$ of our interest,\footnote{The region ${\cal R}_{\mathtt{GW}}^*(P_{XY})$
has some singular points in general, and the following analysis does not apply for those singular points.} i.e.,
\begin{align} \label{eq:slope-of-plane}
\lambda_i^\star = \lambda_i^\star(P_{XY}) := - \frac{\partial}{\partial r_i} R(r_1,r_2|P_{XY}) \bigg|_{\bm{r}=\bm{r}^*}
\end{align}
is well defined for $i=1,2$, where $\bm{r}^* = (r_1^*,r_2^*)$.
Note that $\lambda_i^\star \ge 0$. In the following, we assume that they are strictly positive. In other words, we consider
a boundary point such that $r_0^* > 0$.

For given $P_{W|XY} \in {\cal P}({\cal W}|{\cal X}\times{\cal Y})$, 
$P_{\bar{W}} \in {\cal P}({\cal W})$, 
$P_{\hat{X}|\hat{W}} \in {\cal P}({\cal X}|{\cal W})$, and $P_{\hat{Y}|\hat{W}} \in {\cal P}({\cal Y}|{\cal W})$, 
we introduce the following function:
\begin{align}
& F(P_{W|XY},P_{\bar{W}}, P_{\hat{X}|\hat{W}},P_{\hat{Y}|\hat{W}}) \\
 &:= D(P_{W|XY} \| P_{\bar{W}} | P_{XY}) + \lambda_1^\star \san{E}\bigg[ \log \frac{1}{P_{\hat{X}|\hat{W}}(X|W)} - r_1^* \bigg]
  + \lambda_2^\star \san{E}\bigg[ \log \frac{1}{P_{\hat{X}|\hat{W}}(X|W)} - r_2^* \bigg] \\
&= I(W \wedge X,Y) + D(P_W \| P_{\bar{W}}) + \lambda_1^\star\big\{ H(X|W) + D(P_{X|W} \| P_{\hat{X}|\hat{W}}|P_W) - r_1^* \big\} \\
&~~~ + \lambda_2^\star\big\{ H(Y|W) + D(P_{Y|W} \| P_{\hat{Y}|\hat{W}}|P_W) - r_2^* \big\}.
\end{align}
From the second expression, we can find that the following holds: 
\begin{align}
R(r_1^*,r_2^*|P_{XY}) = \min_{P_{\bar{W}}} \min_{P_{\hat{X}|\hat{W}}} \min_{P_{\hat{Y}|\hat{W}}} \min_{P_{W|XY}} F(P_{W|XY},P_{\bar{W}}, P_{\hat{X}|\hat{W}},P_{\hat{Y}|\hat{W}}). 
\end{align}

For given $P_{\bar{W}}$, $P_{\hat{X}|\hat{W}}$, $P_{\hat{Y}|\hat{W}}$, $\lambda_1 >0$, and $\lambda_2 > 0$, let
\begin{align}
& \Lambda(x,y| P_{\bar{W}}, P_{\hat{X}|\hat{W}}, P_{\hat{Y}|\hat{W}}, \lambda_1,\lambda_2) \\
&:= - \log \san{E}\bigg[ \exp\bigg\{ \lambda_1 \bigg( r_1^* - \log \frac{1}{P_{\hat{X}|\hat{W}}(x|\bar{W})} \bigg)
 + \lambda_2 \bigg( r_2^* -  \log \frac{1}{P_{\hat{Y}|\hat{W}}(y|\bar{W})} \bigg) \bigg\}  \bigg],
\end{align}
where each term $\exp\{ \cdots \}$ in the expectation is understood as $0$ if either $P_{\hat{X}|\hat{W}}(x|w) = 0$ or
$P_{\hat{Y}|\hat{W}}(y|w) = 0$, and the expectation is taken with respect to $\bar{W} \sim P_{\bar{W}}$.

The following lemma gives a connection between the two functions $F(\cdots)$ and $\Lambda(\cdots)$. 
\begin{lemma} \label{lemma:gray-wyner-tentative-solution}
For any $P_{\bar{W}}$, $P_{\hat{X}|\hat{W}}$, and $P_{\hat{Y}|\hat{W}}$, we have
\begin{align}
\min_{P_{W|XY}} F(P_{W|XY},P_{\bar{W}}, P_{\hat{X}|\hat{W}},P_{\hat{Y}|\hat{W}}) 
 = \san{E}[\Lambda(X,Y| P_{\bar{W}}, P_{\hat{X}|\hat{W}}, P_{\hat{Y}|\hat{W}}, \lambda_1^\star,\lambda_2^\star)],
\end{align}
where the minimization is uniquely achieved by $P_{W|XY}$ such that 
\begin{align}
&P_{W|XY}(w|x,y) \\
&= P_{\bar{W}}(w) \exp\bigg\{ \Lambda(x,y| P_{\bar{W}}, P_{\hat{X}|\hat{W}}, P_{\hat{Y}|\hat{W}}, \lambda_1^\star,\lambda_2^\star) 
 + \lambda_1^\star \bigg( r_1^* -  \log \frac{1}{P_{\hat{X}|\hat{W}}(x|w)} \bigg)
 + \lambda_2^\star \bigg( r_2^* -  \log \frac{1}{P_{\hat{Y}|\hat{W}}(y|w)}  \bigg) \bigg\}
 \label{eq:GW-optimal-tentative}
\end{align}
and $P_{W|XY}(w|x,y) = 0$ whenever either $P_{\hat{X}|\hat{W}}(x|w) =0$ or $P_{\hat{Y}|\hat{W}}(y|w) = 0$.\footnote{The only
exceptional case is where either $P_{\hat{X}|\hat{W}}(x|w) =0$ or $P_{\hat{Y}|\hat{W}}(y|w) = 0$ for every
$w \in \mathtt{supp}(P_{\bar{W}})$. We will not invoke this lemma for such an exceptional case throughout the paper.}
\end{lemma}
\begin{proof}
Without loss of optimality, we can
assume that $P_{W|XY}(w|x,y) =0$ whenever either $P_{\hat{X}|\hat{W}}(x|w) =0$ or $P_{\hat{Y}|\hat{W}}(y|w) = 0$.
Otherwise, the value of $F(P_{W|XY},P_{\bar{W}}, P_{\hat{X}|\hat{W}},P_{\hat{Y}|\hat{W}})$ is infinite.
By using the log-sum inequality (cf.~\cite{csiszar-korner:11}), we have
\begin{align}
& F(P_{W|XY},P_{\bar{W}}, P_{\hat{X}|\hat{W}},P_{\hat{Y}|\hat{W}})  \\
&= \sum_{x,y,w} P_{XY}(x,y) P_{W|XY}(w|x,y) \log \frac{P_{W|XY}(w|x,y)}{P_{\bar{W}}(w)} \\
& + \lambda_1^\star \sum_{x,y,w} P_{XY}(x,y) P_{W|XY}(w|x,y) \bigg\{ \log \frac{1}{P_{\hat{X}|\hat{W}}(x|w)} - r_1^* \bigg\} \\
 & + \lambda_2^\star \sum_{x,y,w} P_{XY}(x,y) P_{W|XY}(w|x,y) \bigg\{ \log \frac{1}{P_{\hat{Y}|\hat{W}}(y|w)} - r_2^* \bigg\} \\
&= \sum_{x,y,w} P_{XY}(x,y) P_{W|XY}(w|x,y) \log \frac{P_{W|XY}(w|x,y)}{P_{\bar{W}}(w) \exp\big\{ \lambda_1^\star \big( r_1^*-  \log \frac{1}{P_{\hat{X}|\hat{W}}(x|w) } \big) 
 + \lambda_2^\star \big( r_2^* - \log \frac{1}{P_{\hat{Y}|\hat{W}}(y|w)} \big) \big\}} \\
&\ge \san{E}[\Lambda(X,Y| P_{\bar{W}}, P_{\hat{X}|\hat{W}}, P_{\hat{Y}|\hat{W}}, \lambda_1^\star,\lambda_2^\star)],
\end{align}
where the equality holds if and only if $P_{W|XY}$ is given by \eqref{eq:GW-optimal-tentative}.
\end{proof}

Let $P_{W|XY}^\star$ be an optimal test channel that achieve $R(r_1^*,r_2^*|P_{XY})$, 
and let $P_{W^\star}$, $P_{X|W^\star}$,
and $P_{Y|W^\star}$ be corresponding output distribution and conditional distributions, respectively. 
Then, note that 
\begin{align}
R(r_1^*,r_2^*|P_{XY})
&= \min_{P_{\bar{W}}} \min_{P_{\hat{X}|\hat{W}}} \min_{P_{\hat{Y}|\hat{W}}} \min_{P_{W|XY}} F(P_{W|XY},P_{\bar{W}}, P_{\hat{X}|\hat{W}},P_{\hat{Y}|\hat{W}}) \\
&\le \min_{P_{W|XY}} F(P_{W|XY}, P_{W^\star}, P_{X|W^\star}, P_{Y|W^\star}) \label{eq:GW-tentative-optimization} \\
&\le F(P_{W|XY}^\star,  P_{W^\star}, P_{X|W^\star}, P_{Y|W^\star}) \\
&= R(r_1^*,r_2^*|P_{XY}). 
\end{align}
This implies that $P_{W|XY}^\star$ achieves the minimization in \eqref{eq:GW-tentative-optimization}.
Thus, by Lemma \ref{lemma:gray-wyner-tentative-solution}, $P_{W|XY}^\star$ must satisfy
\begin{align}
P_{W|XY}^\star(w|x,y) 
&= P_{W^\star}(w) \exp\bigg\{ \Lambda(x,y| P_{W^\star}, P_{X|W^\star}, P_{Y|W^\star}, \lambda_1^\star,\lambda_2^\star) \\
&~~~~~~~~~ + \lambda_1^\star \bigg( r_1^* -  \log \frac{1}{P_{X|W^\star}(x|w)} \bigg)
 + \lambda_2^\star \bigg( r_2^* -  \log \frac{1}{P_{Y|W^\star}(y|w)}  \bigg) \bigg\},
\end{align}
which is equivalent to 
\begin{align}
& \Lambda(x,y| P_{W^\star}, P_{X|W^\star}, P_{Y|W^\star}, \lambda_1^\star,\lambda_2^\star) \\
&= \log \frac{P_{W|XY}^\star(w|x,y)}{P_{W^\star}(w)} + \lambda_1^\star \bigg( \log \frac{1}{P_{X|W^\star}(x|w)} - r_1^* \bigg)
 + \lambda_2^\star \bigg( \log \frac{1}{P_{Y|W^\star}(y|w)} - r_2^* \bigg)
 \label{eq:GW-explicit-form-of-tilted}
\end{align}
for $w \in \mathtt{supp}(P_{W|XY}^\star(\cdot|x,y))$.

Although optimal test channels may not be unique,\footnote{In fact, since the auxiliary alphabet does not have any
semantic meaning, for a given optimal test channel, we can always produce another optimal
test channel by permuting symbols in ${\cal W}$.} we have the following important property.
\begin{lemma} \label{lemma:GW-well-defined}
Let $P_{W_1|XY}^\star$ and $P_{W_2|XY}^\star$ be optimal test channels, and let
$P_{W_i^\star}$, $P_{X|W_i^\star}$, and $P_{Y|W_i^\star}$ for $i=1,2$ be corresponding output distribution
and conditional distributions. Then, we have
\begin{align} 
\Lambda(x,y| P_{W_1^\star}, P_{X|W_1^\star}, P_{Y|W_1^\star},\lambda_1^\star, \lambda_2^\star)
= \Lambda(x,y| P_{W_2^\star}, P_{X|W_2^\star}, P_{Y|W_2^\star},\lambda_1^\star, \lambda_2^\star).
 \label{eq:GW-equivalence-tilted-for-two-channels}
\end{align}
\end{lemma}
\begin{proof}
Let ${\cal T} = \{1,2\}$ be time-sharing alphabet, and let 
$\tilde{R}(r_1,r_2|P_{XY})$ be defined in the same manner as $R(r_1,r_2|P_{XY})$ where the auxiliary alphabet
is extended from ${\cal W}$ to ${\cal T} \times {\cal W}$. In fact, by the support lemma (cf.~\cite{csiszar-korner:11}), 
this extension does not
change the value, i.e., 
\begin{align}
\tilde{R}(r_1,r_2|P_{XY}) = R(r_1,r_2|P_{XY}).
\end{align}
We also note that, for $\alpha_1,\alpha_2>0$ with $\alpha_1 + \alpha_2 = 1$,  a test channel $P_{TW|XY}^\star$ defined by 
\begin{align} \label{eq:GW-time-shared-test-channel}
P_{TW|XY}^\star(t,w|x,y) := \alpha_t P_{W_t|XY}^\star(w|x,y)
\end{align}
is an optimal test channel for $\tilde{R}(r_1^*,r_2^*|P_{XY})$. Thus, we can apply the same argument that leads to
\eqref{eq:GW-explicit-form-of-tilted}, and we conclude that the value of
\begin{align}
\log \frac{P_{TW|XY}^\star(t,w|x,y)}{P_{T^\star W^\star}(t,w)} + \lambda_1^\star \bigg( \log \frac{1}{P_{X|T^\star W^\star}(x|t,w)} - r_1^* \bigg)
 + \lambda_2^\star \bigg( \log \frac{1}{P_{Y|T^\star W^\star}(y|t,w)} - r_2^* \bigg)
\end{align}
does not depend on $(t,w) \in \mathtt{supp}(P_{TW|XY}^\star(\cdot|x,y))$, where $P_{T^\star W^\star}$,
$P_{X|T^\star W^\star}$, and $P_{Y|T^\star W^\star}$ are output distribution and conditional distributions
induced by $P_{TW|XY}^\star$.
This together with \eqref{eq:GW-time-shared-test-channel} imply that,\footnote{Also note that
$\mathtt{supp}(P_{TW|XY}^\star(\cdot|x,y)) = \{1\} \times \mathtt{supp}(P_{W_1|XY}^\star(\cdot|x,y)) \cup \{2\} \times \mathtt{supp}(P^\star_{W_2|XY}(\cdot|x,y))$} for any $w \in \mathtt{supp}(P_{W_1|XY}^\star(\cdot|x,y))$ and $w^\prime \in \mathtt{supp}(P_{W_2|XY}^\star(\cdot|x,y))$,
\begin{align}
& \log \frac{P_{W_1|XY}^\star(w|x,y)}{P_{W_1^\star}(w)} + \lambda_1^\star \bigg( \log \frac{1}{P_{X|W_1^\star}(x|w)} - r_1^* \bigg)
 + \lambda_2^\star \bigg( \log \frac{1}{P_{Y|W_1^\star}(y|w)} - r_2^* \bigg) \\
&=  \log \frac{P_{W_2|XY}^\star(w^\prime|x,y)}{P_{W_2^\star}(w^\prime)} + \lambda_1^\star \bigg( \log \frac{1}{P_{X|W_2^\star}(x|w^\prime)} - r_1^* \bigg)
 + \lambda_2^\star \bigg( \log \frac{1}{P_{Y|W_2^\star}(y|w^\prime)} - r_2^* \bigg).
\end{align} 
Thus, we have \eqref{eq:GW-equivalence-tilted-for-two-channels}.
\end{proof}

Because of Lemma \ref{lemma:GW-well-defined}, the following tilted information density is 
well-defined irrespective of the choice of optimal test channels.\footnote{\label{footnote:uniqueness-tilted}
Because of \cite[Theorem 2.4.2]{berger:71},
the $d$-tilted information in the rate-distortion problem is also defined irrespective of the choice of optimal test channels;
the minimum and the maximum in \cite[Remark 9]{kostina:12} are superfluous. See also \cite[Remark in p.~69]{csiszar:74b}.}
\begin{definition}
Let $P_{W|XY}^\star$ be an optimal test channel, and let $P_{W^\star}$, $P_{X|W^\star}$, and $P_{Y|W^\star}$
be corresponding output distribution and conditional distributions, respectively. Then, the tilted information density
for Gray-Wyner network (with respect to $r_0$-axis) is defined by
\begin{align}
\jmath_{XY}(x,y) = \jmath_{XY}(x,y|r_1^*,r_2^*) := \Lambda(x,y|P_{W^\star},P_{X|W^\star},P_{Y|W^\star},\lambda_1^\star,\lambda_2^\star). 
\end{align}
\end{definition}

From \eqref{eq:GW-explicit-form-of-tilted}, we have
\begin{align}
\jmath_{XY}(x,y) = \log \frac{P_{W|XY}^\star(w|x,y)}{P_{W^\star}(w)} + \lambda_1^\star\bigg( \log \frac{1}{P_{X|W^\star}(x|w)} - r_1^* \bigg)
 + \lambda_2^\star \bigg( \log \frac{1}{P_{Y|W^\star}(y|w)} - r_2^* \bigg)
 \label{eq:GW-explicit-form-of-tilted-information}
\end{align}
for every $w \in \mathtt{supp}(P_{W|XY}^\star(\cdot|x,y))$,\footnote{In contrast to
the property of $d$-tilted information in \cite[Eq.~(17)]{kostina:12}, 
\eqref{eq:GW-explicit-form-of-tilted-information} only holds for $w \in \mathtt{supp}(P_{W|XY}^\star(\cdot|x,y))$
instead of $w \in \mathtt{supp}(P_{W^\star})$. This stems from the fact that either $\log \frac{1}{P_{X|W^\star}(x|w)}$
or $\log \frac{1}{P_{Y|W^\star}(y|w)}$ may be infinite, while
distortion measures that may take infinity are excluded in \cite{kostina:12}.
} and thus
\begin{align}
\san{E}[\jmath_{XY}(X,Y)] = R(r_1^*,r_2^*|P_{XY}).
\end{align}

Now we consider to differentiate $R(r_1^*,r_2^*|P_\theta)$ with respect to $\theta$
around $\xi = \theta(P_{XY})$ (cf.~Section \ref{subsection:notation} for the parametric notation).
For a technical reason, 
we assume that there exists a set of optimal test channels
$\{ P_{W|X_\theta Y_\theta} \}_{\theta \in {\cal N}(\xi)}$ around a neighbour ${\cal N}(\xi)$ of  $\xi$ such that 
$P_{W|X_\theta Y_\theta} $ is differentiable.\footnote{In fact, this regularity condition can be replaced by 
any regularity condition that guarantees the validity of \eqref{eq:derivative-coincide-0}.}  
Also, we take ${\cal N}(\xi)$ sufficiently small so that 
\begin{align} \label{eq:gray-wyner-support-equivalent-assumption}
\mathtt{supp}( P_{W|X_\xi Y_\xi}^\star(\cdot|x,y)) \subset \mathtt{supp}( P_{W|X_\theta Y_\theta}^\star(\cdot|x,y))
\end{align}
for every $x,y$.
The following lemma can be proved in a similar manner as  \cite[Theorem 2.2]{kostina:thesis}.
\begin{lemma} \label{lemma:GW-derivative}
Let $\xi = \theta(P_{XY})$. Then, we have
\begin{align}
\frac{\partial R(r_1^*,r_2^*|P_\theta)}{\partial \theta_i} \bigg|_{\theta = \xi} = \jmath_{X Y}(i) - \jmath_{X Y}(m).
\end{align}
\end{lemma}
\begin{proof}
Let $(X_\theta,Y_\theta) \sim P_\theta$. Since we can write 
\begin{align}
R(r_1^*,r_2^* | P_\theta) = \sum_{x,y} P_\theta(x,y) \jmath_{X_\theta Y_\theta}(x,y),
\end{align}
we have
\begin{align}
\frac{\partial R(r_1^*,r_2^* |P_\theta)}{\partial \theta_i} \bigg|_{\theta=\xi}
&= \sum_{x,y} \frac{\partial P_\theta(x,y)}{\partial \theta_i} \bigg|_{\theta=\xi} \jmath_{X_\xi Y_\xi}(x,y)
 + \sum_{x,y} P_\xi(x,y) \frac{\partial \jmath_{X_\theta Y_\theta}(x,y)}{\partial \theta_i} \bigg|_{\theta=\xi} \\
&= \jmath_{X_\xi Y_\xi}(i) - \jmath_{X_\xi Y_\xi}(m) + \sum_{x,y} P_\xi(x,y) \frac{\partial \jmath_{X_\theta Y_\theta}(x,y)}{\partial \theta_i} \bigg|_{\theta=\xi}.
\end{align}
We now evaluate the second term as follows. By \eqref{eq:GW-explicit-form-of-tilted-information} and the assumption \eqref{eq:gray-wyner-support-equivalent-assumption},
we have
\begin{align}
\san{E}[ \jmath_{X_\theta Y_\theta}(X_\xi, Y_\xi)]
&= \sum_{w,x,y} P_\xi(x,y) P_{W|X_\xi Y_\xi}^\star(w|x,y) \jmath_{X_\theta Y_\theta}(x,y) \\
&= \sum_{w,x,y} P_\xi(x,y) P_{W|X_\xi Y_\xi}^\star(w|x,y) \bigg[ \log \frac{P_{W|X_\theta Y_\theta}^\star(w|x,y)}{P_{W_\theta^\star}(w)} \\
& ~~~+ \lambda_{1,\theta}^\star \bigg(\log \frac{1}{P_{X_\theta|W_\theta^\star}(x|w)} - r_1^* \bigg)
  + \lambda_{2,\theta}^\star \bigg(\log \frac{1}{P_{Y_\theta|W_\theta^\star}(y|w)} - r_2^* \bigg)  \bigg],
\end{align}
where $\lambda_{i,\theta}$ is defined by \eqref{eq:slope-of-plane} for $P_\theta$.
Thus, we have\footnote{In the following calculation, the base of logarithm is $e$ instead of $2$, which is irrelevant to the final answer.}
\begin{align}
\sum_{x,y} P_\xi(x,y) \frac{\partial \jmath_{X_\theta Y_\theta}(x,y)}{\partial \theta_i} \bigg|_{\theta=\xi} 
&= \frac{\partial \san{E}[ \jmath_{X_\theta Y_\theta}(X_\xi, Y_\xi)]}{\partial \theta_i} \bigg|_{\theta = \xi} \\
&= \frac{\partial \san{E}[\log P_{W|X_\theta Y_\theta}^\star(W_\xi^\star|X_\xi,Y_\xi)] }{\partial \theta_i} \bigg|_{\theta=\xi}
 - \frac{\partial \san{E}[\log P_{W_\theta^\star}(W_\xi^\star)]}{\partial \theta_i} \bigg|_{\theta =\xi} \\
&~~~+ \frac{\partial \lambda_{1,\theta}^\star}{\partial \theta_i} \bigg|_{\theta=\xi} \bigg( H(X_\xi|W_\xi^\star) - r_1^* \bigg)
 - \lambda_{1,\xi}^\star \frac{\partial \san{E}[\log P_{X_\theta|W_\theta^\star}(X_\xi|W_\xi^\star)]}{\partial \theta_i} \bigg|_{\theta=\xi}\\
&~~~+ \frac{\partial \lambda_{2,\theta}^\star}{\partial \theta_i} \bigg|_{\theta=\xi} \bigg( H(Y_\xi|W_\xi^\star) - r_2^* \bigg)
 - \lambda_{2,\xi}^\star \frac{\partial \san{E}[\log P_{Y_\theta|W_\theta^\star}(Y_\xi|W_\xi^\star)]}{\partial \theta_i} \bigg|_{\theta=\xi} \\
&= \frac{\partial}{\partial \theta_i} \san{E}\bigg[ \frac{P_{W|X_\theta Y_\theta}^\star(W_\xi^\star|X_\xi,Y_\xi)}{P_{W|X_\xi Y_\xi}^\star(W_\xi^\star|X_\xi,Y_\xi)} \bigg] \bigg|_{\theta=\xi}
 - \frac{\partial }{\partial \theta_i} \san{E}\bigg[ \frac{P_{W_\theta^\star}(W_\xi^\star)}{P_{W_\xi^\star}(W_\xi^\star)} \bigg] \bigg|_{\theta=\xi} \\
 &~~~ - \lambda_{1,\xi}^\star \frac{\partial}{\partial \theta_i} \san{E}\bigg[ \frac{P_{X_\theta|W_\theta^\star}(X_\xi|W_\xi^\star)}{P_{X_\xi|W_\xi^\star}(X_\xi|W_\xi^\star)} \bigg] \bigg|_{\theta=\xi} 
 - \lambda_{2,\xi}^\star \frac{\partial}{\partial \theta_i} \san{E}\bigg[ \frac{P_{Y_\theta|W_\theta^\star}(Y_\xi|W_\xi^\star)}{P_{Y_\xi|W_\xi^\star}(Y_\xi|W_\xi^\star)} \bigg] \bigg|_{\theta=\xi} \\
&= 0, \label{eq:derivative-coincide-0}
\end{align}
where the third equality follows from $H(X_\xi|W_\xi^\star) = r_1^*$ and $H(Y_\xi|W_\xi^\star) = r_2^*$.
\end{proof}

\section{Coding Theorem} \label{section:coding-theorem}

In this section, we characterize the second-order region of the Gray-Wyner network.
We first describe the statement, and then it will be proved in 
Sections \ref{subsection:achievability} and \ref{subsection:converse}.

\begin{theorem} \label{theorem:main}
For a given boundary point $(r_0^*,r_1^*,r_2^*) \in {\cal R}_{\mathtt{GW}}(P_{XY})$,
suppose that the function $R(r_1,r_2|P_\theta)$ defined by \eqref{eq:optimal-r0-function} is twice differentiable 
with respect to $(r_1,r_2,\theta)$ around $(r_1^*,r_2^*,\theta(P_{XY}))$ and those second derivatives are bounded.
Also, a regularity condition for Lemma \ref{lemma:GW-derivative} is satisfied. Then, we have
\begin{align}
{\cal L}_{\mathtt{GW}}(\varepsilon;r_0^*,r_1^*,r_2^*) 
 = \big\{ (L_0,L_1,L_2) : L_0 + \lambda_1^\star L_1 + \lambda_2^\star L_2 \ge \sqrt{V_{XY}} \san{Q}^{-1}(\varepsilon) \big\}
\end{align}
for $0<\varepsilon < 1$, where $\lambda_i^\star$ is given by \eqref{eq:slope-of-plane} and 
\begin{align}
V_{XY} := \san{V}\big[ \jmath_{XY}(X,Y) \big].
\end{align}
\end{theorem}

\subsection{Proof of Achievability} \label{subsection:achievability}

In this section, we prove the achievability part of Theorem \ref{theorem:main}.
For each type $P_{\bar{X}\bar{Y}} \in {\cal P}_n({\cal X}\times {\cal Y})$, we pick a conditional type
$P_{\bar{W}|\bar{X}\bar{Y}} \in {\cal P}_n({\cal W}|{\cal X}\times{\cal Y})$, and then construct 
a code ${\cal C}_n \subset {\cal T}_{\bar{W}}^n$ such that, for every $(\bm{x},\bm{y}) \in {\cal T}_{\bar{X}\bar{Y}}^n$,
there exists $\bm{w} \in {\cal C}_n$ satisfying $(\bm{w},\bm{x},\bm{y}) \in {\cal T}_{\bar{W}\bar{X}\bar{Y}}^n$.
Basic strategy is the same as the covering lemma in the rate distortion
(cf.~\cite{YuSpe:93} and \cite[Chapter 9]{csiszar-korner:11}).

\begin{lemma} \label{lemma:code-construction}
Suppose that $n \ge n_0(|{\cal X}|,|{\cal Y}|,|{\cal W}|)$.
Given type $P_{\bar{X}\bar{Y}} \in {\cal P}_n({\cal X}\times {\cal Y})$ and any test channel
$P_{W|\bar{X}\bar{Y}}$ (not necessarily conditional type), there exists a conditional type
$P_{\bar{W}|\bar{X}\bar{Y}}$ satisfying
\begin{align} \label{eq:truncated-condition}
\big| P_{\bar{W}|\bar{X}\bar{Y}}(w|x,y) - P_{W|\bar{X}\bar{Y}}(w|x,y) \big| \le \frac{1}{n P_{\bar{X}\bar{Y}}(x,y)}
\end{align}
for every $(x,y) \in \mathtt{supp}(P_{\bar{X}\bar{Y}})$ and $w \in \mathtt{supp}(P_{W|\bar{X}\bar{Y}}(\cdot|x,y))$,
and a subset ${\cal C}_n \subset {\cal T}_{\bar{W}}^n$ such that 
\begin{align}
|{\cal C}_n| \le \exp\big\{ n I(\bar{W} \wedge \bar{X},\bar{Y}) + (|{\cal X}||{\cal Y}||{\cal W}| + 4) \log(n+1) \big\}
\end{align}
and such that, for any $(\bm{x},\bm{y}) \in {\cal T}_{\bar{X}\bar{Y}}^n$, there exists $\bm{w} \in {\cal C}_n$ 
satisfying $(\bm{w},\bm{x},\bm{y}) \in {\cal T}_{\bar{W}\bar{X}\bar{Y}}^n$.
\end{lemma}
\begin{proof}
By truncating the given test channel $P_{W|\bar{X}\bar{Y}}$ into conditional type, we can obtain 
conditional type $P_{\bar{W}|\bar{X}\bar{Y}}$ satisfying \eqref{eq:truncated-condition}. 
Let $Z^{m_n} = \{ Z_1,\ldots,Z_{m_n}\}$ be i.i.d. and uniform over ${\cal T}_{\bar{W}}^n$.
We will show
\begin{align}
\san{E}\bigg[ \sum_{(\bm{x},\bm{y}) \in {\cal T}_{\bar{X}\bar{Y}}^n} 
 \bol{1}[ (Z_i,\bm{x},\bm{y}) \notin {\cal T}_{\bar{W}\bar{X}\bar{Y}}^n~\forall 1 \le i \le m_n ] \bigg] < 1,
 \label{eq:goal-covering}
\end{align}
which implies that there exists ${\cal C}_n$ with $|{\cal C}_n|\le m_n$ satisfying the desired property.
The lefthand side can be manipulated as
\begin{align}
& \sum_{(\bm{x},\bm{y}) \in {\cal T}_{\bar{X}\bar{Y}}^n} 
 \san{E}\bigg[ \bol{1}[ (Z_i,\bm{x},\bm{y}) \notin {\cal T}_{\bar{W}\bar{X}\bar{Y}}^n~\forall 1 \le i \le m_n ] \bigg]  \\
&=  \sum_{(\bm{x},\bm{y}) \in {\cal T}_{\bar{X}\bar{Y}}^n}
 \bigg( 1 - \frac{|{\cal T}_{\bar{W}|\bar{X}\bar{Y}}^n(\bm{x},\bm{y})|}{|{\cal T}_{\bar{W}}^n|} \bigg)^{m_n} \\
&\le \sum_{(\bm{x},\bm{y}) \in {\cal T}_{\bar{X}\bar{Y}}^n} 
 \exp\bigg\{ - \frac{|{\cal T}_{\bar{W}|\bar{X}\bar{Y}}^n(\bm{x},\bm{y})|}{|{\cal T}_{\bar{W}}^n|} m_n \bigg\},
\end{align}
where the last inequality follows from $(1-t)^m \le \exp\{- t m\}$ for $0 < t < 1$.
Furthermore, it can be upper bounded as (cf.~\cite[Lemma 2.5]{csiszar-korner:11})
\begin{align}
& \sum_{(\bm{x},\bm{y}) \in {\cal T}_{\bar{X}\bar{Y}}^n} 
 \exp\bigg\{ - \frac{|{\cal T}_{\bar{W}|\bar{X}\bar{Y}}^n(\bm{x},\bm{y})|}{|{\cal T}_{\bar{W}}^n|} m_n \bigg\} \\
&\le \sum_{(\bm{x},\bm{y}) \in {\cal T}_{\bar{X}\bar{Y}}^n} 
 \exp\big\{ - (n+1)^{-|{\cal X}||{\cal Y}||{\cal W}| } 2^{- n I(\bar{W} \wedge \bar{X},\bar{Y})} m_n \big\} \\
&\le \exp\big\{ n H(\bar{X},\bar{Y}) \big\}
 \exp\big\{ - (n+1)^{-|{\cal X}||{\cal Y}||{\cal W}| } 2^{- n I(\bar{W} \wedge \bar{X},\bar{Y})} m_n \big\} \\ 
&\le \exp\big\{ n \log |{\cal X}||{\cal Y}| \big\}
 \exp\big\{ - (n+1)^{-|{\cal X}||{\cal Y}||{\cal W}| } 2^{- n I(\bar{W} \wedge \bar{X},\bar{Y})} m_n \big\}.
\end{align}
Thus, by taking $m_n$ such that
\begin{align}
& \exp\big\{ n I(\bar{W} \wedge \bar{X},\bar{Y}) + (|{\cal X}||{\cal Y}||{\cal W}| +2) \log (n+1) \big\} \\
&\le m_n 
\le \exp\big\{ n I(\bar{W} \wedge \bar{X},\bar{Y}) + (|{\cal X}||{\cal Y}||{\cal W}| +4) \log (n+1) \big\},
\end{align}
\eqref{eq:goal-covering} holds for sufficiently large $n$. 
\end{proof}

\begin{corollary} \label{corollary:code-construction}
Suppose that $n \ge n_0(|{\cal X}|,|{\cal Y}|,|{\cal W}|)$.
Given type $P_{\bar{X}\bar{Y}} \in {\cal P}_n({\cal X}\times {\cal Y})$, there exists a conditional type
$P_{\bar{W}|\bar{X}\bar{Y}}$ satisfying
\begin{align}
n H(\bar{X}|\bar{W}) &\le n r_1 + 2|{\cal X}||{\cal Y}||{\cal W}| \log n, \\
n H(\bar{Y}|\bar{W}) &\le n r_2 + 2|{\cal X}||{\cal Y}||{\cal W}| \log n,
\end{align}
and a subset ${\cal C}_n \subset {\cal T}_{\bar{W}}^n$ such that 
\begin{align} \label{eq:corollary-covering-size}
\log |{\cal C}_n| \le n R(r_1,r_2|P_{\bar{X}\bar{Y}}) + (3|{\cal X}||{\cal Y}||{\cal W}| + 4) \log (n+1),
\end{align}
and such that, for any $(\bm{x},\bm{y}) \in {\cal T}_{\bar{X}\bar{Y}}^n$, there exists $\bm{w} \in {\cal C}_n$ 
satisfying $(\bm{w},\bm{x},\bm{y}) \in {\cal T}_{\bar{W}\bar{X}\bar{Y}}^n$.
\end{corollary}
\begin{proof}
For the given type $P_{\bar{X}\bar{Y}}$, we pick an optimal test channel $P_{W|\bar{X}\bar{Y}}$
that achieve $R(r_1,r_2|P_{\bar{X}\bar{Y}})$. Then, Lemma \ref{lemma:code-construction}
implies that there exists conditional type $P_{\bar{W}|\bar{X}\bar{Y}}$ satisfying \eqref{eq:truncated-condition}
and a subset ${\cal C}_n$ satisfying the desired properties. From \eqref{eq:truncated-condition}, we have
\begin{align}
\big\| P_{\bar{W}\bar{X}\bar{Y}} - P_{W \bar{X}\bar{Y}} \big\|_1 \le \frac{|{\cal X}||{\cal Y}||{\cal W}|}{n},
\end{align}
where $\| \cdot \|_1$ is the variational distance.
Thus, by the continuity of entropy functions (cf.~\cite[Lemma 2.7]{csiszar-korner:11}), we have
\begin{align}
| H(\bar{X}|\bar{W}) - H(\bar{X}|W) |
&\le | H(\bar{W},\bar{X}) - H(W,\bar{X}) | + |H(W) - H(\bar{W}) | \\
&\le \frac{|{\cal X}||{\cal Y}||{\cal W}|}{n} \log \frac{|{\cal X}||{\cal W}|}{\frac{|{\cal X}||{\cal Y}||{\cal W}|}{n}}
 + \frac{|{\cal X}||{\cal Y}||{\cal W}|}{n} \log \frac{|{\cal W}|}{\frac{|{\cal X}||{\cal Y}||{\cal W}|}{n}} \\
&\le \frac{2 |{\cal X}||{\cal Y}||{\cal W}|}{n} \log n,
\end{align}
Similarly, we have
\begin{align}
| H(\bar{Y}|\bar{W}) - H(\bar{Y}|W) | \le \frac{2 |{\cal X}||{\cal Y}||{\cal W}|}{n} \log n
\end{align}
and 
\begin{align}
| I(\bar{W} \wedge \bar{X},\bar{Y}) - I(W \wedge \bar{X},\bar{Y})| 
&\le | H(\bar{W}) - H(W) | + | H(\bar{W},\bar{X},\bar{Y}) - H(W,\bar{X},\bar{Y})| \\
& \le \frac{2 |{\cal X}||{\cal Y}||{\cal W}|}{n} \log n.
\end{align}
\end{proof}

From Corollary \ref{corollary:code-construction}, we can derive the following.
\begin{lemma}
There exists a code $\Phi_n$ such that 
\begin{align} \label{eq:achievability-bound}
\Pe(\Phi_n|P_{XY}^n) \le \Pr\bigg( r_{0,n} < R(r_{1,n}, r_{2,n}| \san{t}_{X^n Y^n} )\bigg),
\end{align}
where $\san{t}_{X^n Y^n}$ is the joint type of $(X^n,Y^n)$, and 
\begin{align}
r_{0,n} &:= \frac{1}{n} \log |{\cal M}_0^{(n)}| - \frac{(4|{\cal X}||{\cal Y}||{\cal W}| + 4) \log (n+1)}{n}, \label{eq:achievability-tentative-rate-0} \\
r_{1,n} &:= \frac{1}{n} \log |{\cal M}_1^{(n)}| - \frac{2|{\cal X}||{\cal Y}||{\cal W}| \log n}{n}, \label{eq:achievability-tentative-rate-1} \\   
r_{2,n} &:= \frac{1}{n} \log |{\cal M}_2^{(n)}| - \frac{2|{\cal X}||{\cal Y}||{\cal W}| \log n}{n}. \label{eq:achievability-tentative-rate-2}
\end{align}
\end{lemma}
\begin{proof}
We consider the following coding scheme.
Upon observing $(X^n,Y^n)$, the encoder first compute its joint type $\san{t}_{X^n Y^n}$,
and sends it to the decoders via the common channel by using
$|{\cal X}||{\cal Y}| \log (n+1)$ bits. Then, for joint type $P_{\bar{X}\bar{Y}} = \san{t}_{X^n Y^n}$, the encoder finds\footnote{The encoder and the decoders
agree on the choice of the test channel $P_{\bar{W}|\bar{X}\bar{Y}}$ and the subset ${\cal C}_n$ for each joint type.}
the test channel $P_{\bar{W}|\bar{X}\bar{Y}}$ and ${\cal C}_n$ that are specified by
Corollary \ref{corollary:code-construction}, where we set $(r_1,r_2)$ in the corollary to be
$(r_{1,n}, r_{2,n})$ given by \eqref{eq:achievability-tentative-rate-1} and \eqref{eq:achievability-tentative-rate-2}, respectively.
If $\log |{\cal C}_n|$ exceeds 
\begin{align} \label{eq:error-event}
\log |{\cal M}_0^{(n)}| - |{\cal X}||{\cal Y}| \log (n+1),
\end{align} 
then the system aborts and declares an error. Otherwise, the encoder send $\bm{w} \in {\cal C}_n$
satisfying $(\bm{w}, X^n,Y^n) \in {\cal T}_{\bar{W}\bar{X}\bar{Y}}^n$ to the decoders via the common channel.
Since $X^n \in {\cal T}_{\bar{X}|\bar{W}}^n(\bm{w})$, the encoder sends index of $X^n$ in ${\cal T}_{\bar{X}|\bar{W}}^n(\bm{w})$
to the first decoder via the first private channel by using 
\begin{align}
\log |{\cal T}_{\bar{X}|\bar{W}}^n(\bm{w})|
&\le n H(\bar{X}|\bar{W}) \\
&\le \log |{\cal M}_1^{(n)}|
\end{align}
bits. Similarly, since $Y^n \in {\cal T}_{\bar{Y}|\bar{W}}^n(\bm{w})$, the encoder sends the index of $Y^n$
in ${\cal T}_{\bar{Y}|\bar{W}}^n(\bm{w})$ to the second decoder via the second private channel by using
\begin{align}
\log |{\cal T}_{\bar{Y}|\bar{W}}^n(\bm{w})|
&\le n H(\bar{Y}|\bar{W}) \\
&\le \log |{\cal M}_2^{(n)}|
\end{align}
bits. 

In the above coding scheme, the error occurs only when $\log |{\cal C}_n|$ exceeds \eqref{eq:error-event}. 
Thus, noting \eqref{eq:achievability-tentative-rate-0} and \eqref{eq:corollary-covering-size},
the error probability is upper bounded by
\begin{align} 
\Pr\bigg( r_{0,n} < R(r_{1,n}, r_{2,n}| \san{t}_{X^n Y^n} )\bigg),
\end{align}
which completes the proof.
\end{proof}

Now, we evaluate the right hand side of \eqref{eq:achievability-bound}.
Let (cf.~Section \ref{subsection:notation} for the notations $\theta_i$ and $m$)
\begin{align}
{\cal K}_n := \bigg\{ P_{\bar{X}\bar{Y}} \in {\cal P}_n({\cal X}\times {\cal Y}) : |\theta_i(P_{\bar{X}\bar{Y}}) - \theta_i(P_{XY})| 
  \le \sqrt{\frac{\log n}{n}}~~\forall 1 \le i \le m-1 \bigg\}.
\end{align}
The following is an immediate consequence of Hoeffding's inequality.
\begin{proposition}
We have
\begin{align}
\Pr\bigg( \san{t}_{X^n Y^n} \notin {\cal K}_n \bigg) \le \frac{2(m-1)}{n^2}.
\end{align}
\end{proposition}

To evaluate \eqref{eq:achievability-bound}, we proceed as follows. We set
\begin{align}
\frac{1}{n} \log |{\cal M}_i^{(n)}| = r_i^* + \frac{L_i}{\sqrt{n}}
\end{align}
for $i=1,2$ (we will specify $|{\cal M}_0^{(n)}|$ later). 
Since we assumed that the second order derivatives of $R(r_1,r_2|P_\theta)$ with respect to
$(r_1,r_2,\theta)$ are bounded around a neighbor of $(r_1^*,r_2^*,\theta(P_{XY}))$, and since
$r_{i,n} - r_i^* = \frac{L_i}{\sqrt{n}} + O\left(\frac{\log n}{n}\right)$,  when $\san{t}_{X^n Y^n} \in {\cal K}_n$, we 
can Taylor expand $R(r_{1,n},r_{2,n}|\san{t}_{X^n Y^n})$ as
\begin{align}
R(r_{1,n},r_{2,n}|\san{t}_{X^n Y^n}) 
&\le R(r_1^*, r_2^*|P_{XY}) - \lambda_1^\star \frac{L_1}{\sqrt{n}}  
 - \lambda_2^\star \frac{L_2}{\sqrt{n}} \\
&~~~ + \sum_{i=1}^{m-1} (\theta_i(\san{t}_{X^nY^n}) - \theta_i(P_{XY})) 
  \big( \jmath_{XY}(i) - \jmath_{XY}(m) \big) + c \frac{\log n}{n} \\
&= R(r_1^*, r_2^*|P_{XY}) - \lambda_1^\star \frac{L_1}{\sqrt{n}} 
 - \lambda_2^\star \frac{L_2}{\sqrt{n}} \\
&~~~ + \sum_{x,y} ( \san{t}_{X^n Y^n}(x,y) - P_{XY}(x,y)) \jmath_{XY}(x,y) + c \frac{\log n}{n} \\
&=  \sum_{x,y} \san{t}_{X^n Y^n}(x,y) \jmath_{XY}(x,y) - \lambda_1^\star \frac{L_1}{\sqrt{n}} 
 - \lambda_2^\star \frac{L_2}{\sqrt{n}} + c \frac{\log n}{n} \\
&= \frac{1}{n} \sum_{i=1}^n \jmath_{XY}(X_i,Y_i)  - \lambda_1^\star \frac{L_1}{\sqrt{n}} 
 - \lambda_2^\star \frac{L_2}{\sqrt{n}} + c \frac{\log n}{n}
\end{align}
for some constant $c>0$ provided that $n$ is sufficiently large, where the first inequality follows from
Lemma \ref{lemma:GW-derivative} and the second equality follows from 
$\san{E}[\jmath_{XY}(X,Y)] = R(r_1^*, r_2^*|P_{XY})$.
Thus, we have
\begin{align}
& \Pr\bigg( r_{0,n} < R(r_{1,n}, r_{2,n}| \san{t}_{X^n Y^n} )\bigg) \\
&\le \Pr\bigg( \san{t}_{X^n Y^n} \in {\cal K}_n,~r_{0,n} < R(r_{1,n}, r_{2,n}| \san{t}_{X^n Y^n}) \bigg)
 +  \Pr\bigg( \san{t}_{X^n Y^n} \notin {\cal K}_n \bigg) \\
&\le \Pr\bigg( r_{0,n} + \lambda_1^\star \frac{L_1}{\sqrt{n}} + \lambda_2^\star \frac{L_2}{\sqrt{n}} 
 < \frac{1}{n} \sum_{i=1}^n \jmath_{XY}(X_i,Y_i) + c \frac{\log n}{n} \bigg) + \frac{2(m-1)}{n^2}.
\end{align}
Thus, if we set 
\begin{align}
\frac{1}{n} \log |{\cal M}_0^{(n)}| = R(r_1^*,r_2^*|P_{XY}) +  \frac{L_0}{\sqrt{n}} + \frac{(4|{\cal X}||{\cal Y}||{\cal W}| + 4 + c) \log (n+1)}{n}, 
\end{align}
there exists a code $\Phi_n$ such that 
\begin{align}
\Pe(\Phi_n|P_{XY}^n)
&\le \Pr\bigg( R(r_1^*,r_2^*|P_{XY}) +  \frac{L_0}{\sqrt{n}} + \lambda_1^\star \frac{L_1}{\sqrt{n}} + \lambda_2^\star \frac{L_2}{\sqrt{n}} 
 < \frac{1}{n} \sum_{i=1}^n \jmath_{XY}(X_i,Y_i) \bigg) + \frac{2(m-1)}{n^2} \\
&= \Pr\bigg( L_0 + \lambda_1^\star L_1 + \lambda_2^\star L_2 
 < \frac{1}{\sqrt{n}} \bigg( \sum_{i=1}^n \jmath_{XY}(X_i,Y_i) - n R(r_1^*,r_2^*|P_{XY}) \bigg) \bigg)  + \frac{2(m-1)}{n^2}.
\end{align}
Thus, if we set $L_0,L_1,L_2$ so that 
\begin{align}
L_0 + \lambda_1^\star L_1 + \lambda_2^\star L_2 \ge \sqrt{V_{XY}} \san{Q}^{-1}(\varepsilon),
\end{align}
by applying the central limit theorem, we have
\begin{align}
\limsup_{n\to\infty} \Pe(\Phi_n|P_{XY}^n) \le \varepsilon.
\end{align}
Since this code also satisfies \eqref{eq:definition-second-achievability-0}-\eqref{eq:definition-second-achievability-2}, we have shown
$(\varepsilon,r_0^*,r_1^*,r_2^*)$-achievability of $(L_0,L_1,L_2)$. \qed

\subsection{Proof of Converse} \label{subsection:converse}

In this section, we prove the converse part of Theorem \ref{theorem:main}.
We first derive a kind of strong converse bound when a code $\Phi_n$ is applied to source $(X^n,Y^n) \sim P_{{\cal T}_{\bar{X}\bar{Y}}^n}$
for the uniform distribution $P_{{\cal T}_{\bar{X}\bar{Y}}^n}$ on the type class for a fixed type $P_{\bar{X}\bar{Y}}$. 

\begin{lemma} \label{lemma:fixed-type-converse}
Suppose that the correct probability satisfies
\begin{align} \label{eq:lower-bound-on-correct}
\Pc(\Phi_n | P_{{\cal T}_{\bar{X}\bar{Y}}^n}) \ge 2^{- n \alpha_n}
\end{align}
for some positive number $\alpha_n$. 
Let $\beta_n$ be another positive number.
Then there exists $P_{\bar{W}|\bar{X}\bar{Y}}$ with $|{\cal W}| \le |{\cal X}| |{\cal Y}| +2$ such that 
\begin{align}
\frac{1}{n} \log |{\cal M}_0^{(n)}| &\ge I(\bar{W} \wedge \bar{X},\bar{Y}) - \frac{|{\cal X}||{\cal Y}| \log (n+1)}{n} - (\alpha_n+\beta_n), \\
\frac{1}{n} \log |{\cal M}_1^{(n)}| &\ge H(\bar{X}|\bar{W}) - \frac{1}{n}  - 2^{-n \beta_n} \log |{\cal X}|, \\
\frac{1}{n} \log |{\cal M}_2^{(n)}| &\ge H(\bar{Y}|\bar{W}) - \frac{1}{n} - 2^{- n \beta_n} \log |{\cal Y}|,
\end{align}
where $(\bar{X},\bar{Y}) \sim P_{\bar{X}\bar{Y}}$.
\end{lemma}
\begin{proof}
We prove this lemma by using the perturbation approach used in \cite{GuEff:09, GuEff:11}.
Let 
\begin{align}
{\cal D}_{\bar{X}\bar{Y}} := \bigg\{ (\bm{x},\bm{y}) \in {\cal T}_{\bar{X}\bar{Y}}^n: \psi_1(\varphi_0(\bm{x},\bm{y}),\varphi_1(\bm{x},\bm{y})) = \bm{x},   
 \psi_2(\varphi_0(\bm{x},\bm{y}),\varphi_2(\bm{x},\bm{y})) = \bm{y}  \bigg\}.
\end{align}
be the set of correctly decodable sequences on ${\cal T}_{\bar{X}\bar{Y}}^n$.
Let $Q_{{\cal T}_{\bar{X}\bar{Y}}^n}$ be a distribution on ${\cal T}_{\bar{X}\bar{Y}}^n$ defined by
\begin{align}
Q_{{\cal T}_{\bar{X}\bar{Y}}^n}(\bm{x},\bm{y}) = \frac{2^{n(\alpha_n +\beta_n)} P_{{\cal T}_{\bar{X}\bar{Y}}^n}(\bm{x},\bm{y})}{2^{n(\alpha_n +\beta_n)} P_{{\cal T}_{\bar{X}\bar{Y}}^n}({\cal D}_{\bar{X}\bar{Y}}) + (1 - P_{{\cal T}_{\bar{X}\bar{Y}}^n}({\cal D}_{\bar{X}\bar{Y}}))}
\end{align}
for $(\bm{x},\bm{y}) \in {\cal D}_{\bar{X}\bar{Y}}$ and
\begin{align}
Q_{{\cal T}_{\bar{X}\bar{Y}}^n}(\bm{x},\bm{y}) = \frac{ P_{{\cal T}_{\bar{X}\bar{Y}}^n}(\bm{x},\bm{y})}{2^{n(\alpha_n +\beta_n)} P_{{\cal T}_{\bar{X}\bar{Y}}^n}({\cal D}_{\bar{X}\bar{Y}}) + (1 - P_{{\cal T}_{\bar{X}\bar{Y}}^n}({\cal D}_{\bar{X}\bar{Y}}))}
\end{align}
for $(\bm{x},\bm{y}) \notin {\cal D}_{\bar{X}\bar{Y}}$. Then, from \eqref{eq:lower-bound-on-correct}, we have
\begin{align}
Q_{{\cal T}_{\bar{X}\bar{Y}}^n}( {\cal D}_{\bar{X}\bar{Y}} ) \ge \frac{2^{n \beta_n}}{2^{n \beta_n} + 1}.
\end{align}
In other words, if we use the same code $\Phi_n$ to source $(X^n,Y^n) \sim Q_{{\cal T}_{\bar{X}\bar{Y}}^n}$, we have
\begin{align} \label{eq:error-probability-perturbed}
\Pe(\Phi_n | Q_{{\cal T}_{\bar{X}\bar{Y}}^n}) \le 2^{- n \beta_n}.
\end{align}
Furthermore, for every $(\bm{x},\bm{y}) \in {\cal T}_{\bar{X}\bar{Y}}^n$, we have 
\begin{align}
Q_{{\cal T}_{\bar{X}\bar{Y}}^n}(\bm{x},\bm{y})
 \le 2^{n(\alpha_n +\beta_n)} P_{{\cal T}_{\bar{X}\bar{Y}}^n}(\bm{x},\bm{y})
  \label{eq:upper-bound-perturbated}
\end{align}
and
\begin{align}
Q_{{\cal T}_{\bar{X}\bar{Y}}^n}(\bm{x},\bm{y})
 &\ge \frac{ P_{{\cal T}_{\bar{X}\bar{Y}}^n}(\bm{x},\bm{y})}{2^{n(\alpha_n +\beta_n)} P_{{\cal T}_{\bar{X}\bar{Y}}^n}({\cal D}_{\bar{X}\bar{Y}}) +  2^{n(\alpha_n +\beta_n)} (1 - P_{{\cal T}_{\bar{X}\bar{Y}}^n}({\cal D}_{\bar{X}\bar{Y}}))} \\
  &= 2^{- n(\alpha_n +\beta_n)} P_{{\cal T}_{\bar{X}\bar{Y}}^n}(\bm{x},\bm{y}).
  \label{eq:lower-bound-perturbated}
\end{align}

Now, by a slight modification of the standard argument\footnote{Note that all the 
information quantities are evaluated with respect to $(X^n,Y^n) \sim Q_{{\cal T}_{\bar{X}\bar{Y}}^n}$; for example, 
$(X_i,Y_i)$ and $(X^{i-1},Y^{i-1})$ may not be independent.}, we have
\begin{align}
\frac{1}{n} \log |{\cal M}_0^{(n)}|
&\ge \frac{1}{n} H(S_0) \\
&\ge \frac{1}{n} I(S_0 \wedge X^n,Y^n) \\
&= \frac{1}{n} \sum_{i=1}^n I(S_0 \wedge X_i, Y_i | X^{i-1}, Y^{i-1}) \\
&= \frac{1}{n} \sum_{i=1}^n \bigg[ I(S_0, X^{i-1},Y^{i-1} \wedge X_i,Y_i) - I(X^{i-1},Y^{i-1} \wedge X_i, Y_i) \bigg] \\
&= \frac{1}{n} \sum_{i=1}^n I(S_0, X^{i-1},Y^{i-1} \wedge X_i,Y_i) - \frac{1}{n}\bigg[ \sum_{i=1}^n H(X_i,Y_i) - H(X^n,Y^n) \bigg] \\
&= \frac{1}{n} \sum_{i=1}^n I(W_i \wedge X_i, Y_i) - \frac{1}{n}\bigg[ \sum_{i=1}^n H(X_i,Y_i) - H(X^n,Y^n) \bigg] \\
&= I(W_J \wedge X_J, Y_J | J) - \bigg[ H(X_J, Y_J|J) - \frac{1}{n} H(X^n, Y^n) \bigg] \\
&= I(W_J \wedge X_J, Y_J | J) + I(J \wedge X_J, Y_J) - \bigg[ H(X_J, Y_J) - \frac{1}{n} H(X^n, Y^n) \bigg] \\
&= I(J, W_J \wedge X_J, Y_J) - \bigg[ H(X_J, Y_J) - \frac{1}{n} H(X^n, Y^n) \bigg], 
 \label{eq:enc-0-lower-bound}
\end{align}
where $S_0 = \varphi_0(X^n,Y^n)$, $W_i = (S_0, X^{i-1},Y^{i-1})$, and $J$ is the uniform 
random variable on $\{1,\ldots,n\}$ that is independent of all the other random variables\footnote{Note that
$J$ and $(X_J,Y_J)$ may not be independent.}. We also have
\begin{align}
\frac{1}{n} \log |{\cal M}_1^{(n)}|
&\ge \frac{1}{n} H(S_1) \\
&\ge \frac{1}{n} H(S_1|S_0) \\
&\ge \frac{1}{n} I(S_1 \wedge X^n | S_0) \\
&= \frac{1}{n} H(X^n | S_0) - \frac{1}{n} H(X^n | S_0, S_1) \\
&\ge \frac{1}{n} H(X^n| S_0) - \bigg[ \frac{1}{n} + 2^{- n \beta_n} \log |{\cal X}| \bigg] \\
&= \frac{1}{n} \sum_{i=1}^n H(X_i | S_0, X^{i-1}) - \bigg[ \frac{1}{n} + 2^{- n \beta_n} \log |{\cal X}| \bigg] \\
&\ge \frac{1}{n} \sum_{i=1}^n H(X_i | S_0, X^{i-1}, Y^{i-1}) - \bigg[ \frac{1}{n} + 2^{- n \beta_n} \log |{\cal X}| \bigg] \\
&= H(X_J | J, W_J) - \bigg[ \frac{1}{n} + 2^{- n \beta_n} \log |{\cal X}| \bigg],
 \label{eq:enc-1-lower-bound}
\end{align}
where $S_1 = \varphi_1(X^n,Y^n)$, and the forth inequality follows from the Fano inequality and \eqref{eq:error-probability-perturbed}. Similarly, we have
\begin{align}
\frac{1}{n} \log |{\cal M}_2^{(n)}|
\ge H(Y_J | J, W_J) - \bigg[ \frac{1}{n} + 2^{- n \beta_n} \log |{\cal Y}| \bigg].
 \label{eq:enc-2-lower-bound}
\end{align}

By the support lemma (cf.~\cite{csiszar-korner:11}), there exists $P_{W|X_J Y_J}$ such that 
$|{\cal W}| \le |{\cal X}| |{\cal Y}| + 2$ and 
\begin{align}
I(J, W_J \wedge X_J, Y_J) &= I(W \wedge X_J, Y_J), \\
H(X_J | J, W_J) &= H(X_J|W), \\
H(Y_J | J, W_J) &= H(Y_J|W).
\end{align}

Now, we claim that the distribution $P_{X_J Y_J}$ coincides with the type $P_{\bar{X}\bar{Y}}$. 
In fact, for every fixed $(\bm{x},\bm{y}) \in {\cal T}_{\bar{X}\bar{Y}}^n$, we have
\begin{align}
\Pr\bigg( (X_J, Y_J) = (a,b) \bigg| (X^n,Y^n) = (\bm{x},\bm{y}) \bigg) = P_{\bar{X}\bar{Y}}(a,b)
\end{align}
for every $(a,b) \in {\cal X} \times {\cal Y}$.
Thus, we have
\begin{align}
P_{X_J Y_J}(a,b)
&= \sum_{(\bm{x},\bm{y}) \in {\cal T}_{\bar{X}\bar{Y}}^n} Q_{{\cal T}_{\bar{X}\bar{Y}}^n}(\bm{x},\bm{y}) \Pr\bigg( (X_J, Y_J) = (a,b) \bigg| (X^n,Y^n) = (\bm{x},\bm{y}) \bigg) \\
&= P_{\bar{X}\bar{Y}}(a,b).
\end{align}
Here, by letting $\bar{W}$ as the random variable induced by the channel $P_{W|X_J Y_J}$
from $(\bar{X},\bar{Y}) \sim P_{\bar{X}\bar{Y}}$, we have
\begin{align}
I(J, W_J \wedge X_J, Y_J) &= I(\bar{W} \wedge \bar{X}, \bar{Y}), \label{eq:cardinality-bounded-0} \\
H(X_J | J, W_J) &= H(\bar{X}|\bar{W}), \label{eq:cardinality-bounded-1} \\
H(Y_J | J, W_J) &= H(\bar{Y}|\bar{W}). \label{eq:cardinality-bounded-2}
\end{align} 

Finally, we evaluate the residual term in \eqref{eq:enc-0-lower-bound}.
Since $P_{X_J Y_J} = P_{\bar{X}\bar{Y}}$, we have
\begin{align} \label{eq:equivalence-of-entropy}
H(X_J,Y_J) = H(\bar{X},\bar{Y}).
\end{align}
Furthermore, it is well known that (cf.~\cite{csiszar-korner:11})
\begin{align} \label{eq:type-class-bound}
n H(\bar{X},\bar{Y}) - |{\cal X}| |{\cal Y}| \log (n+1) 
 \le \log |{\cal T}_{\bar{X}\bar{Y}}^n| 
 \le n H(\bar{X},\bar{Y}).
\end{align}
From \eqref{eq:upper-bound-perturbated} and \eqref{eq:lower-bound-perturbated}, we have
\begin{align}
\bigg| \log \frac{1}{Q_{{\cal T}_{\bar{X}\bar{Y}}^n}(\bm{x},\bm{y})} - \log|{\cal T}_{\bar{X}\bar{Y}}^n| \bigg| \le n (\alpha_n + \beta_n)
\end{align}
for every $(\bm{x},\bm{y}) \in {\cal T}_{\bar{X}\bar{Y}}^n$. Thus, we have
\begin{align} \label{eq:entropy-perturbated-bound}
\bigg| H(X^n,Y^n) - \log |{\cal T}_{\bar{X}\bar{Y}}^n| \bigg| \le n (\alpha_n + \beta_n).
\end{align}
By combining \eqref{eq:equivalence-of-entropy}, \eqref{eq:type-class-bound}, and \eqref{eq:entropy-perturbated-bound}, we have
\begin{align}
& \bigg| H(X_J,Y_J) - \frac{1}{n} H(X^n,Y^n) \bigg| \\
&\le \bigg| H(X_J,Y_J) - \frac{1}{n} \log|{\cal T}_{\bar{X}\bar{Y}}^n| \bigg|
 + \bigg| \frac{1}{n} \log|{\cal T}_{\bar{X}\bar{Y}}^n| - \frac{1}{n} H(X^n,Y^n) \bigg| \\
&\le \frac{|{\cal X}||{\cal Y}| \log (n+1)}{n} + (\alpha_n+\beta_n).
\label{eq:residual-bound}
\end{align}

Consequently, from \eqref{eq:enc-0-lower-bound}, \eqref{eq:enc-1-lower-bound},
\eqref{eq:enc-2-lower-bound}, \eqref{eq:cardinality-bounded-0}-\eqref{eq:cardinality-bounded-2}, 
and \eqref{eq:residual-bound}, we have the claim of the lemma.
\end{proof}

We now use Lemma \ref{lemma:fixed-type-converse} by setting
\begin{align} \label{eq:choice-of-alpha-beta}
\alpha_n = \beta_n = \frac{\log n}{n}.
\end{align}
\begin{lemma}
For any code $\Phi_n$, it holds that 
\begin{align} \label{eq:converse-error-lower-bound-with-type}
\Pe(\Phi_n|P_{XY}^n) \ge \Pr\bigg( r_{0,n}  < R(r_{1,n}, r_{2,n}| \san{t}_{X^n Y^n}) \bigg) - \frac{1}{n},
\end{align}
where $\san{t}_{X^n Y^n}$ is the joint type of $(X^n,Y^n)$, 
\begin{align}
r_{0,n} &;= \frac{1}{n} \log |{\cal M}_0^{(n)}| + \frac{|{\cal X}||{\cal Y}| \log (n+1)}{n} + (\alpha_n + \beta_n), \\
r_{1,n} &:=  \frac{1}{n} \log |{\cal M}_1^{(n)}| + \frac{1}{n} + 2^{-n \beta_n} \log |{\cal X}|, \\
r_{2,n} &:= \frac{1}{n} \log |{\cal M}_2^{(n)}| + \frac{1}{n} + 2^{-n \beta_n} \log |{\cal Y}|,
\end{align}
and $\alpha_n$ and $\beta_n$ are set as \eqref{eq:choice-of-alpha-beta}.
\end{lemma}
\begin{proof}
Let $\bm{r}_n = (r_{0,n},r_{1,n},r_{2,n})$.
Lemma \ref{lemma:fixed-type-converse} implies that if
\begin{align}
\bm{r}_n \notin {\cal R}_{\mathtt{GW}}(P_{\bar{X}\bar{Y}}),
\end{align}
then the correct probability satisfies 
\begin{align}
\Pc(\Phi_n| P_{{\cal T}_{\bar{X}\bar{Y}}^n}) < 2^{- n \alpha_n}.
\end{align}
Thus, we have
\begin{align}
\Pe(\Phi_n|P_{XY}^n)
&= \sum_{P_{\bar{X}\bar{Y}} \in {\cal P}_n({\cal X} \times {\cal Y})} P_{XY}^n({\cal T}_{\bar{X}\bar{Y}}^n) \Pe(\Phi_n|P_{{\cal T}_{\bar{X}\bar{Y}}^n}) \\
&\ge \sum_{P_{\bar{X}\bar{Y}} \in {\cal P}_n({\cal X} \times {\cal Y}) \atop \bm{r}_n \notin {\cal R}_{\mathtt{GW}}(P_{\bar{X}\bar{Y}})}
 P_{XY}^n({\cal T}_{\bar{X}\bar{Y}}^n) (1 - 2^{- n \alpha_n}) \\
&\ge \sum_{P_{\bar{X}\bar{Y}} \in {\cal P}_n({\cal X} \times {\cal Y}) \atop \bm{r}_n \notin {\cal R}_{\mathtt{GW}}(P_{\bar{X}\bar{Y}})}
 P_{XY}^n({\cal T}_{\bar{X}\bar{Y}}^n) - \frac{1}{n},
\end{align}
where the last inequality follows from the choice of $\alpha_n$.
By denoting the type of $(X^n,Y^n)$ by $\san{t}_{X^nY^n}$, the first term of the above bound can be written as
\begin{align}
\sum_{P_{\bar{X}\bar{Y}} \in {\cal P}_n({\cal X} \times {\cal Y}) \atop \bm{r}_n \notin {\cal R}_{\mathtt{GW}}(P_{\bar{X}\bar{Y}})}
 P_{XY}^n({\cal T}_{\bar{X}\bar{Y}}^n) 
&= \Pr\bigg( \bm{r}_n \notin {\cal R}_{\mathtt{GW}}(\san{t}_{X^nY^n}) \bigg) \\
 & = \Pr\bigg( r_{0,n} < R(r_{1,n}, r_{2,n}| \san{t}_{X^n Y^n}) \bigg),
\end{align}
which completes the proof.
\end{proof}

Now, we  evaluate the first term of the right hand side of \eqref{eq:converse-error-lower-bound-with-type}.
Suppose that $|{\cal M}_i^{(n)}|$ for $i=0,1,2$ satisfy \eqref{eq:definition-second-achievability-0}-\eqref{eq:definition-second-achievability-2}
for some $(L_0,L_1,L_2)$ satisfying
\begin{align} \label{eq:converse-violation}
L_0 + \lambda_1^\star L_1 + \lambda_2^\star L_2 < \sqrt{V_{XY}} \san{Q}^{-1}(\varepsilon).
\end{align}
Then, we can write
\begin{align} 
\frac{1}{n} \log |{\cal M}_i^{(n)}| = r_i^* + \frac{L_i}{\sqrt{n}} + \delta_{i,n},~~i=0,1,2,
\end{align}
for some $\delta_{i,n} = o(1/\sqrt{n})$. 
Thus, in the same manner as the achievability part, we have\footnote{Note also that $r_0^* = R(r_1^*,r_2^*|P_{XY})$.}
\begin{align}
& \Pr\bigg( r_{0,n} < R(r_{1,n}, r_{2,n}| \san{t}_{X^n Y^n} )\bigg) \\
&\ge \Pr\bigg( \san{t}_{X^nY^n} \in {\cal K}_n,~ r_{0,n} < R(r_{1,n}, r_{2,n}| \san{t}_{X^n Y^n} )\bigg) \\
&\ge \Pr\bigg( \san{t}_{X^nY^n} \in {\cal K}_n,~R(r_1^*,r_2^*|P_{XY}) + \lambda_1^\star \frac{L_1}{\sqrt{n}} + \lambda_2^\star \frac{L_2}{\sqrt{n}} 
 < \frac{1}{n} \sum_{i=1}^n \jmath_{XY}(X_i,Y_i) - \delta_n \bigg) \\
&\ge \Pr\bigg( R(r_1^*,r_2^*|P_{XY}) + \lambda_1^\star \frac{L_1}{\sqrt{n}} + \lambda_2^\star \frac{L_2}{\sqrt{n}} 
 < \frac{1}{n} \sum_{i=1}^n \jmath_{XY}(X_i,Y_i) - \delta_n \bigg) - \Pr\bigg( \san{t}_{X^nY^n} \notin {\cal K}_n \bigg) \\
&\ge \Pr\bigg( R(r_1^*,r_2^*|P_{XY}) + \lambda_1^\star \frac{L_1}{\sqrt{n}} + \lambda_2^\star \frac{L_2}{\sqrt{n}} 
 < \frac{1}{n} \sum_{i=1}^n \jmath_{XY}(X_i,Y_i) - \delta_n \bigg) - \frac{2(m-1)}{n^2}
\end{align}
for some $\delta_n = o(1/\sqrt{n})$. Thus, by the central limit theorem, we have
\begin{align}
\liminf_{n\to\infty} \Pe(\Phi_n|P_{XY}^n) > \varepsilon,
\end{align}
which implies that any $(L_0,L_1,L_2)$ satisfying \eqref{eq:converse-violation} is not 
$(\varepsilon, r_0^*,r_1^*,r_2^*)$-achievable. \qed

\section{On The Pangloss Plane} \label{section:pangloss}

In general, it is extremely difficult to compute the first-order region ${\cal R}_{\mathtt{GW}}^*(P_{XY})$, and 
so do the second-order region ${\cal L}_{\mathtt{GW}}(\varepsilon;r_0^*,r_1^*,r_2^*)$. Nevertheless, to get some insight, let us consider the following tractable case.

The region ${\cal R}_{\mathtt{GW}}^*(P_{XY})$ is contained in the outer region 
characterized by three planes (cf.~Fig.~\ref{Fig:example}):
\begin{align}
r_0 + r_1 + r_2 &\ge H(X,Y), \\
r_0 + r_1 &\ge H(X), \\
r_0 + r_2 &\ge H(Y).
\end{align}
The first plane
is called the {\em Pangloss} plane in \cite{GraWyn:74}. Let 
\begin{align}
{\cal H}(P_{XY}) &:= \big\{ (r_0,r_1,r_2) \in {\cal R}_{\mathtt{GW}}(P_{XY}) : r_0 + r_1 + r_2 = H(X,Y) \big\} \\
 &= \big\{ (I(W \wedge X,Y), H(X|W), H(Y|W)): |{\cal W}| \le |{\cal X}||{\cal Y}| + 2,~ X \markov W \markov Y \big\}
\end{align}
be the set of all achievable rate triplets on the Pangloss plane, where $X \markov W \markov Y$ means $(X,W,Y)$ form Markov chain.
Although explicit characterization of ${\cal H}(P_{XY})$ is not clear in general, it is broader than the following triangular region
\begin{align}
\mathtt{conv}\big\{ (H(X,Y), 0, 0),~(H(Y), H(X|Y),0),~(H(X),0,H(Y|X)) \big\},
\end{align}
and the altitude of the lowermost points is $r_0 = C_{\mathtt{W}}(P_{XY})$, where 
\begin{align}
C_{\mathtt{W}}(P_{XY}) &:= \min\big\{ r_0: \exists r_1,r_2 \mbox{ s.t. } (r_0,r_1,r_2) \in {\cal H}(P_{XY}) \big\} \\
&= \min\big\{ I(W \wedge X,Y) : |{\cal W}| \le |{\cal X}||{\cal Y}|,~X \markov W \markov Y \big\}
\end{align}
is Wyner's common information \cite{wyner:75b} (cf.~Fig.~\ref{Fig:example}).

When $(r_0^*,r_1^*,r_2^*) \in {\cal H}(P_{XY})$, it holds that
\begin{align} \label{eq:first-order-for-pangloss}
R(r_1^*,r_2^*|P_{XY}) = H(X,Y) - r_1^* - r_2^*.
\end{align}
Thus, $\lambda_1^\star = \lambda_2^\star = 1$.
Also, since $(r_0^*,r_1^*,r_2^*)$ is achieved by an optimal test channel satisfying $X \markov W^\star \markov Y$, 
it holds that
\begin{align}
\jmath_{XY}(x,y) &= \log \frac{P_{W|XY}^\star(w|x,y)}{P_{W^\star}(w)} 
 + \lambda_1^\star\bigg( \log \frac{1}{P_{X|W^\star}(x|w)} - r_1^* \bigg)
 + \lambda_2^\star\bigg( \log \frac{1}{P_{Y|W^\star}(y|w)} - r_2^* \bigg) \\
&= \log \frac{1}{P_{XY}(x,y)} - r_1^* - r_2^*.
\end{align}
Thus, the second-order region is characterized as follows:\footnote{It is apparent from \eqref{eq:first-order-for-pangloss} 
that the second derivative of $R(r_1,r_2|P_\theta)$ around $(r_1^*,r_2^*,\theta(P_{XY}))$ is bounded. Furthermore, instead of checking differentiability of
test channels, we can directly differentiate $\jmath_{XY}(x,y)$ in this case, and thus the 
validity of Lemma \ref{lemma:GW-derivative} is also guaranteed.}
\begin{corollary}
When $(r_0^*,r_1^*,r_2^*)$ is an strict inner point of ${\cal H}(P_{XY})$,\footnote{Since we do not know an explicit form of
${\cal R}_{\mathtt{GW}}(P_{XY})$ outside ${\cal H}(P_{XY})$, it is not clear if the regularity condition is satisfied or not on
the boundary of ${\cal H}(P_{XY})$.} it holds that
\begin{align}
{\cal L}_{\mathtt{GW}}(\varepsilon; r_0^*,r_1^*,r_2^*) 
 = \big\{ (L_0,L_1,L_2) : L_0 + L_1 + L_2 \ge \sqrt{V_{XY}} \san{Q}^{-1}(\varepsilon) \big\},
\end{align}
where $V_{XY}$ is given by
\begin{align}
V_{XY} &= \san{V}\big[ \jmath_{XY}(X,Y) \big] \\
&= \san{V}\bigg[ \log \frac{1}{P_{XY}(X,Y)} \bigg].
\end{align}
\end{corollary}
In fact, the sum constraint on the second-order rates in the above corollary coincides with 
the cooperative outer bound, where the two decoders cooperate. 
Thus, on the Pangloss plane, there is no
sum-rate loss compared to cooperative decoding scheme up to the second-order,
which is quite remarkable.
However, it does not mean that the auxiliary random variable is not needed;
the auxiliary random variable is needed to construct a code that achieve the optimal
second-order region.

\begin{figure}[tb]
\centering{
\begin{minipage}{.5\textwidth}
\includegraphics[width=\textwidth]{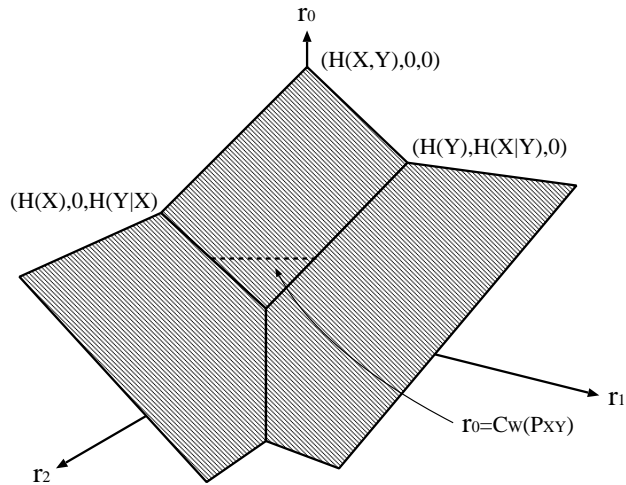}
\caption{A description of an outer region of ${\cal R}_{\mathtt{GW}}(P_{XY})$.}
\label{Fig:example}
\end{minipage}
}
\end{figure}

\section{Discussion} \label{section:discussion}

In this paper, we derived a characterization of the second-order region
of the Gray-Wyner network. Apart from the interest on this network itself, there is another motivation
to study this problem. As we mentioned earlier, the characterization of the first-order region of multi-terminal problems
typically involve auxiliary random variables; involvement of auxiliary random variables is 
one of reasons that the second-order analysis of multi-terminal problems is difficult.
Thus, the result of this paper is an important step toward extending the second-order analysis to multi-terminal problems. 
 
It seems that the next simple problems that involve auxiliary random variables are the coding problems with side-information (cf.~\cite{watanabe:13e}).
In contrast to the Gray-Wyner network, the coding problems with side-information involve Markov chain structures on auxiliary random variables that 
stem from the distributed coding nature of the problems. Thus, the techniques used in this paper are not enough to solve these problems.  
However, we believe that the result in this paper at least gives some hints to tackle those problems. 


\section*{Acknowledgement}

The author would like to thank Yasutada Oohama and Vincent Y. F. Tan 
for valuable comments. The author also would like to thank Victoria Kostina for
valuable comments; in particular, the reference \cite[Remark in p.~69]{csiszar:74b} in
Footnote \ref{footnote:uniqueness-tilted} was informed to the author from her.


\bibliographystyle{../../09-04-17-bibtex/IEEEtranS}
\bibliography{../../09-04-17-bibtex/reference.bib}

\end{document}